\let\latexaddtocontents\addtocontents
\documentclass[a4paper,twocolumn,accepted=2023-07-26]{quantumarticle}
\let\addtocontents\latexaddtocontents
\pdfoutput=1

\usepackage[T1]{fontenc}
\usepackage[utf8]{inputenc}
\usepackage[english]{babel}
\usepackage{csquotes}
\usepackage{amssymb}
\usepackage{amsmath}
\usepackage{amsthm}
\usepackage{enumitem}
\usepackage{mathtools}
\usepackage{mathrsfs}
\usepackage{dsfont}
\usepackage{color}
\usepackage{braket}

\makeatletter
\@addtoreset{equation}{myequation}
\makeatother

\usepackage{hyperref}
\usepackage{cleveref}

\newcommand{\IN}{{\mathbb{N}}}
\newcommand{\IZ}{{\mathbb{Z}}}

\newcommand{\abs}[1]{\left| #1 \right|}

\renewcommand{\d}{\,\mathrm{d}}

\renewcommand{\sc}[1]{\left< #1 \right>}
\newcommand{\nn}[1]{\left\Vert #1 \right\Vert}

\newcommand{\tr}{\operatorname{tr}}

\newcommand{\Id}{\mathds{1}}

\newcommand{\Def}{\mathcal{D}}

\newtheoremstyle{fgrs} 
                        {0.5em}    
                        {0.5em}    
                        {}         
                        {}         
                        {\bfseries}
                        {}        
                        {\newline} 
                        {}
\newtheoremstyle{mydef} 
                        {0.5em}    
                        {0.5em}    
                        {}         
                        {}         
                        {\bfseries}
                        {}        
                        {\newline} 
                        {}
\theoremstyle{mydef}

\AfterEndEnvironment{defn}{\noindent\ignorespaces}
\AfterEndEnvironment{example}{\noindent\ignorespaces}
\AfterEndEnvironment{rem}{\noindent\ignorespaces}
\AfterEndEnvironment{prop}{\noindent\ignorespaces}
\AfterEndEnvironment{thm}{\noindent\ignorespaces}
\AfterEndEnvironment{lemma}{\noindent\ignorespaces}
\AfterEndEnvironment{coro}{\noindent\ignorespaces}
\AfterEndEnvironment{hyp}{\noindent\ignorespaces}
\AfterEndEnvironment{proof}{\noindent\ignorespaces}

\renewcommand{\H}{\mathcal{H}}

\newcommand{\FF}{\mathfrak{F}}

\newcommand{\fin}{{\operatorname{fin}}}

\renewcommand{\d}{{\mathrm d}}

\newcommand{\Fc}{\mathcal{F}}

\newcommand{\ind}{\mathds{1}}

\newcommand{\x}{{\hat{\mathsf x}}}
\newcommand{\expec}[2]{\sc{#1}_{#2}}

\newcommand{\hh}{\mathfrak{h}}

\newtheorem{thm}{Theorem}
\newtheorem{lemma}[thm]{Lemma}
\newtheorem{prop}[thm]{Proposition}

\theoremstyle{definition}

\theoremstyle{remark}
\newtheorem{remark}[thm]{Remark}

\crefname{thm}{Theorem}{Theorems}
\crefname{coro}{Corollary}{Corollaries}
\crefname{hyp}{Hypothesis}{Hypotheses}
\Crefname{hyp}{Hypothesis}{Hypotheses}
\crefname{lemma}{Lemma}{Lemmas}
\Crefname{lemma}{Lemma}{Lemmas}
\crefname{prop}{Proposition}{Propositions}
\crefname{enumi}{}{}
\Crefname{enumi}{}{}
\creflabelformat{enumi}{#2(#1)#3}
\crefname{equation}{}{}
\Crefname{equation}{}{}

\usepackage[style=ieee,giveninits=true,maxnames=4, backend=biber,url=false,isbn=false]{biblatex}
\newbibmacro{string+doi}[1]{%
	\iffieldundef{doi}{#1}{\href{https://doi.org/\thefield{doi}}{#1}}}
\DeclareFieldFormat{title}{\usebibmacro{string+doi}{\mkbibemph{#1}}}
\DeclareFieldFormat[article]{title}{\usebibmacro{string+doi}{\mkbibquote{#1}}}
\DeclareFieldFormat[inproceedings]{title}{\usebibmacro{string+doi}{\mkbibquote{#1}}}
\DeclareFieldFormat[incollection]{title}{\usebibmacro{string+doi}{\mkbibquote{#1}}}

\newcommand{\Sp}[1]{\mathcal T_1^+(#1)}
\newcommand{\N}{\mathcal{N}}
\newcommand{\T}{\mathcal{T}}

\newcommand{\ABb}{\partial_{AB}}

\allowdisplaybreaks
\begin{document}
\title{Thermal Area Law for Lattice Bosons}

\author{Marius Lemm}
\email{marius.lemm@uni-tuebingen.de}
\author{Oliver Siebert}
\email{oliver.siebert@uni-tuebingen.de}
\affiliation{Department of Mathematics, University of Tübingen, Auf der Morgenstelle 10, 
72076 T\"ubingen, Germany}

 \begin{abstract}
	A physical system is said to satisfy a thermal area law if the mutual information between two adjacent regions in the Gibbs state is controlled by the area of their boundary.  Lattice bosons have recently gained significant interest because they can be precisely tuned in experiments and bosonic codes can be employed in quantum error correction to circumvent classical no-go theorems. However, the proofs of many basic information-theoretic inequalities such as the thermal area law break down for bosons because their interactions are unbounded.
	  Here, we rigorously derive a thermal area law for a class of bosonic Hamiltonians in any dimension which includes the paradigmatic Bose-Hubbard model. The main idea to go beyond bounded interactions is to introduce a quasi-free reference state with artificially decreased chemical potential by means of a double Peierls-Bogoliubov estimate. 
		
\end{abstract}

\maketitle

\section{Introduction}

In quantum many-body systems with translation-invariant short-ranged interactions the entanglement entropy of the ground state typically satisfies an \textit{area law} -- meaning that it is bounded by a constant times the boundary surface area of $A$ (as opposed to the trivial bound which would entail the volume of $A$). The area law captures our physical intuition that correlations are concentrated on short distances and therefore only occur across the boundary cut. It is extremely useful in practice as it severely restricts the admissible many-body states for approximating ground states (i.e., quantum matter) and can thus serve to overcome the notorious \textit{curse of dimensionality} through the famous density matrix renormalization group (DMRG) numerical algorithm \cite{white1993density,dmrg_higher1,dmrg_higher3,dmrg_higher2}. The connection is clearest for 1D lattice systems, where a state satisfies an area law if and only if it is representable as a matrix product state (MPS) with fixed bond dimension independent of the system size \cite{dalzell2019locally,verstraete2006matrix}. For detailed reviews also covering the higher-dimensional situation, see \cite{eisert2010colloquium,ge2016area}. 

Area laws for the entanglement entropy as described above have their origins in the holographic principle in the context of quantum gravity \cite{bousso2002holographic} and have been numerically observed in a large number of many-body systems. They have been derived for gapped 1D spin systems  \cite{hastings_arealaw,arad2012improved,arad2013area,arad2017rigorous,huang2014area}, for 1D quantum states with finite correlation lengths \cite{brandao2013area,cho2018realistic}, gapped harmonic lattice systems \cite{audenaert2002entanglement,
plenio2005entropy,cramer2006correlations,cramer2006entanglement}, ground states in the same gapped phase as others obeying an area law \cite{van2013entanglement,marien2016entanglement}, models whose Hamiltonian spectra satisfy related
conditions \cite{hastings2007entropy,masanes2009area}, certain frustration-free spin systems \cite{de2010solving}, tree-graph systems \cite{abrahamsen2019polynomial}, models exhibiting local topological order \cite{michalakis2012stability} and 
high-dimensional systems under additional assumptions such as frustration-freeness \cite{masanes2009area,brandao2015entanglement,hastings2007entropy,cho2014sufficient,anshu2020entanglement,anshu2022area}. There has also been recent progress for certain long-range interactions \cite{longrange2,kuwahara2020area}. A general statement in higher dimensions remains elusive and this problem is known as the \textit{area law conjecture}.

The analog of the area law for Gibbs states is called the thermal area law. Gibbs states are important for fundamental reasons
 and their efficient simulability  via tensor networks hinges on thermal area laws \cite{improved,renyi_positivetemperature,PRXQuantum.2.040331}. For Gibbs states, the total correlations between two regions $A$ and $B$ are quantified by their \textit{mutual information} \cite{nielsen2002quantum,alhambra_survey,improved}
$$
I(A:B) = S(\rho_A) + S(\rho_B) - S(\rho_{AB} ),
$$
where $S$ denotes the von Neumann entropy, $\rho_{AB}$ the Gibbs state of the full system, and $\rho_A, \rho_B$ the reduced density matrices corresponding to $A$ and $B$, respectively.
At zero temperature, $I(A:B) $ reduces to twice the entanglement entropy. An area law at positive temperature was derived in a seminal work of Wolf et. al \cite{wolf} who proved for local and bounded interactions that
$$
I(A:B)\leq C\beta \abs{\ABb} 
$$
with $\ABb$ being the boundary region between $A$ and $B$, and
whenever $A$ and $B$ are disjoint and together make up the entire lattice.
Very recently, the $\beta$-scaling result was improved to $ \beta^{2/3}$  \cite{improved} which matters for the experimentally relevant regime of low temperatures ($\beta\to\infty$). This dependence is not far from
optimal, since Gottesman and Hastings found a 1D model for which the scaling of the mutual information is at least $\beta^{1/5}$ for large $\beta$  \cite{gottesman2010entanglement}. Moreover, a generalization to various Rényi generalizations of the mutual information was given in \cite{renyi_positivetemperature} and this has important applications to simulating and approximating Gibbs states \cite{PRXQuantum.2.040331}. Further results of thermal area laws were established for free fermions \cite{free_fermions}, for the entanglement negativity (instead of the mutual information) \cite{entanglement_negativity}, as a result of rapid mixing for dissipative quantum lattice systems \cite{rapid_mixing1,rapid_mixing2} (with a logarithmic correction)   and numerically for some spin chains showing a $\log \beta$-dependence \cite{thermal_numerics}. A current account of thermal area laws and related phenomena is given in \cite{alhambra_survey}. There has also been recent progress in the experimental verification of area laws for both the entanglement entropy and the mutual information at positive temperature by means of ultra-cold atom
simulators \cite{tajik2022experimental}.

Thermal area laws provide us with universal properties of Gibbs states independent of the system size. They are especially valuable for lower temperatures, while at high temperatures  many analytic properties are available, e.g. exponential decay of correlations \cite{araki1969gibbs,gross1979decay,park1995uniqueness,ueltschi2004,kliesch2014locality,frohlich2015some}, the large deviation principle \cite{lenci2005large,netovcny2004large,kuwahara2020gaussian} or the approximate quantum Markov property \cite{kato2019quantum,kuwahara2020clustering}. Such characterizations are important for the computation of Gibbs states, in general an NP-hard problem \cite{barahona1982computational,goldberg2015complexity},  which in turn is fundamental in novel applications such as quantum machine learning \cite{amin2018quantum,anshu2021sample}, semidefinite programming solvers \cite{brandao2017quantum,van2020quantum} or the imaginary-time evolution in the framework of near-term quantum devices \cite{nisq1,nisq2,nisq3,nisq4,nisq5,nisq6,nisq7}.
%

In particular, area laws can be used for the approximation of 
Gibbs state by matrix product operators (MPO) and their higher-dimensional analogs \cite{alhambra_survey,improved,PRXQuantum.2.040331,mpo_arealaw_link}. 
 In special cases like one-dimensional quantum spin systems they gave rise to concrete algorithms for the efficient  approximation by MPO, where the temperature behavior in the area law determines the maximal bond dimension \cite{improved}. Such algorithms also  proved to be useful for the representation of Gibbs states as a convex combination of MPS \cite{convex_mps}, or for the approximation of ground states under a low-energy-density assumption typically observed for gapped systems.
It is generally believed that the scaling of the mutual information with $\beta$ in the low temperature regime is related to the computational complexity of the ground space of the models \cite{alhambra_survey}. Another direct way to use a thermal area law is to note that controlling the mutual information automatically controls all standard correlation functions, see \cite{wolf} and \Cref{sec:truncated correlations}
below. 


A critical limitation of the existing results is that they only hold for \textit{bounded} interactions and bounded local Hilbert space dimension. This is naturally the case for quantum spin systems and lattice fermions. However, for lattice \textit{bosons} as described, e.g., by the paradigmatic Bose-Hubbard Hamiltonian, the interactions are \textit{unbounded} and the standard arguments fail. 

There has recently been a surge of interest in bosonic lattice systems and related models for three main reasons: (i) The Hamiltonians can be experimentally fine-tuned for cold atoms in optical lattices \cite{ultracold_gases_review}, which makes them promising platform for quantum simulation and quantum engineering, see also \cite{childs2014bose}. (ii) Bosonic encoding can be used in quantum information processing which can provide multiple advantages over finite-dimensional discrete-variable (DV) codes \cite{bosonic_coding,qec}. (iii) In many cases the standard techniques of quantum information theory fail (including the derivations of the thermal area law) because of the unbounded interactions.  

Area laws for non-interacting bosons were considered in \cite{cramer2006correlations,cramer2006entanglement,cramer2007statistics}, but the case of interacting bosons, including the paradigmatic Bose-Hubbard model proved elusive to rigorous analysis. A partly numerical investigation was given in \cite{bh_es1} with a focus on the phase transition from Mott insulator to superfluid by analogy with other symmetry breaking transitions \cite{bh_es2,bh_es3,bh_es4}. Recently, Abrahamsen et al.\ rigorously proved an area law for gapped ground states of 1D bosonic lattice Hamiltonians in \cite{abrahamsen2022entanglement} by using a truncation of the local Hilbert spaces and a quantum number tail bound from \cite{tong2021provably}. Their work only concerns the zero temperature case and leaves open the positive temperature case. 


In this work we rigorously derive the first thermal area law for a broad class of bosonic Hamiltonians in any dimension including the paradigmatic Bose-Hubbard model. In a nutshell, our result states that
under natural assumptions on the bosonic lattice gases (e.g., short-ranged hopping), we again have the bound
$$
I(A:B)\leq C\beta \abs{\ABb}
$$
for all $\beta\geq 1$. The precise result is  \cref{th:main} below. The $\beta$-scaling is the same as that found by \cite{wolf} and our proof uses the same basic idea as a starting point, namely to use the Gibbs variational principle to bound the mutual information by a difference of boundary energy expectations (\cref{th:lem1}). However, the unbounded interactions then pose technical difficulties which we overcome by introducing a quasi-free reference state with artificially decreased chemical potential by means of a double Peierls-Bogoliubov estimate. This reduces us to computations with quasi-free states that can be completed with Wick's rule. Further details are explained below and in the appendix.

\section{Setup for infinite-dimensional Hilbert spaces}

A technical point in the description of many-boson systems is that the local Hilbert spaces are infinite-dimensional since the particle number is unbounded. For this reason, we include this short preliminary section in which we recall the elegant approach to thermal area laws via the Gibbs variational principle by Wolf et al.\ \cite{wolf} and note that it adapts straightforwardly to the infinite-dimensional situation. These abstract results are then utilized for the Bose-Hubbard model in the following section.


Let $\hh$ be the local Hilbert space for one site, $\Lambda$ a finite set and $A, B \subseteq \Lambda$ such that $A \cap B = \emptyset$ and $A \sqcup B = \Lambda$. We set
\begin{align*}
\H_A = \bigotimes_{x \in A} \hh, \quad \H_B = \bigotimes_{x \in B} \hh, \quad \H_{AB} = \H_A \otimes \H_B.
\end{align*}
We suppose that the Hamiltonian can be decomposed as
\begin{align}
\label{eq:AB partition}
H_{AB} = H_A \otimes \Id  + \Id \otimes H_B + H_\partial
\end{align}

\begin{figure}
	\centering
	\includegraphics[width=6cm]{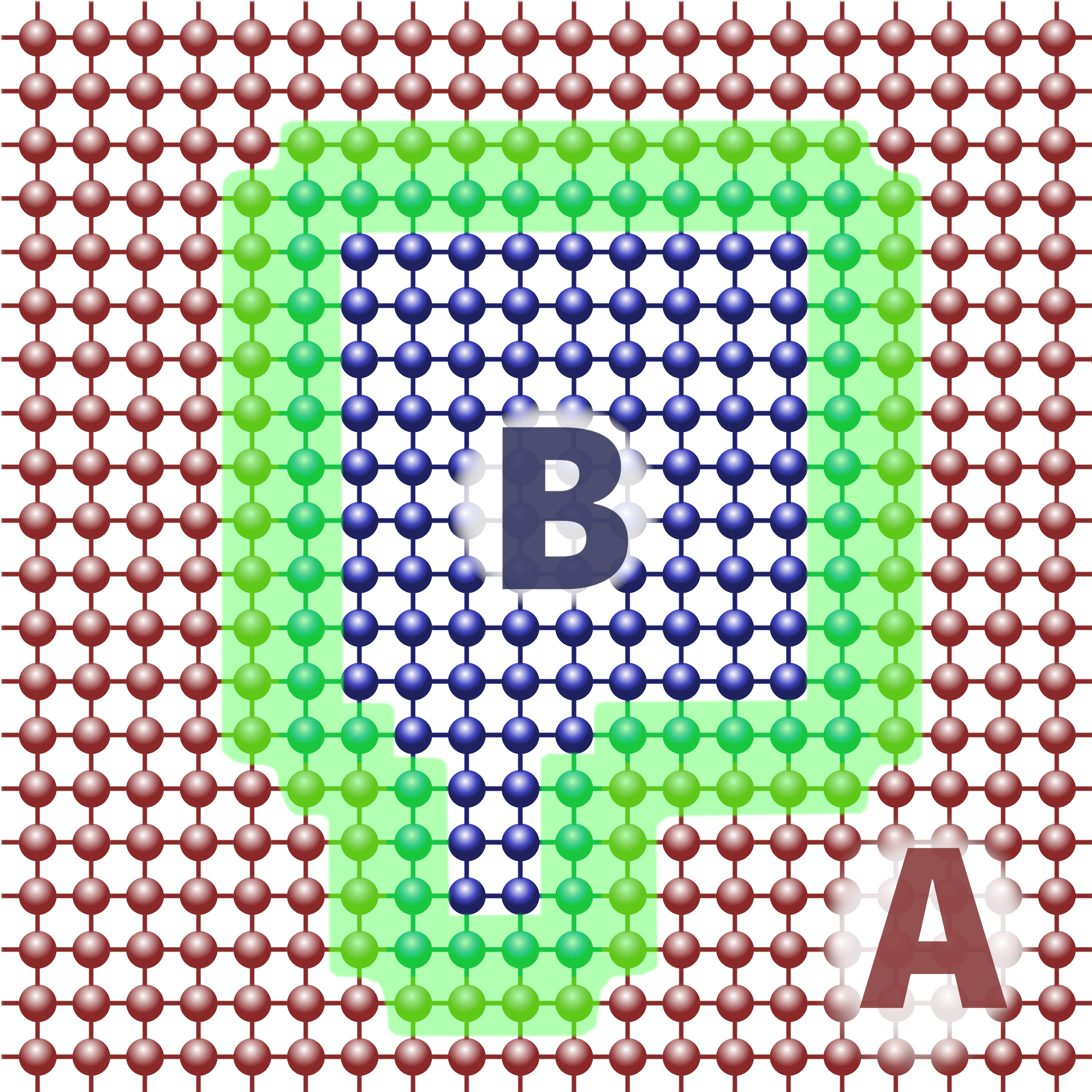}
	\caption{Periodic box for $d=2$ and $L=20$ decomposed into two regions $A$ and $B$. The green region shows the boundary region $\ABb$. The bonds connecting $A$ and $B$ make up the boundary Hamiltonian $H_\partial$.}
	\label{fig:model}
\end{figure}


For details about operator domains, which are relevant because our operators are unbounded, see \Cref{sec:operator domains}. Let $\Sp\H$ denote the set of all density matrices on $\H$. Then we define the free energy as
\[
F^\beta_{AB} = \inf_{\rho \in \Sp{\H_{AB}}}  \Fc^\beta ( \rho),   
\]
where
\begin{align*}
\Fc^\beta ( \rho)  =     \tr( H_{AB} \rho) -  \frac{S(\rho)}{\beta}, \quad
S(\rho) =  - \tr( \rho \ln \rho ).
\end{align*}
The Gibbs state 
\[
\rho_{AB} = \frac{e^{-\beta H_{AB} }}{Z_{AB} } , \quad Z_{AB} =  \tr e^{-\beta H_{AB} },
\]
minimizes the free energy, i.e., $F^\beta_{AB} = \Fc^\beta ( \rho_{AB} )$. (See \cref{th:passivity} for a proof of this in the infinite-dimensional setting.)
We reduce to the $A$, respectively $B$ subsystem by taking partial traces of the Gibbs state
\begin{align*}
\rho_A = \tr_{\H_B} (\rho_{AB}),\qquad \rho_B = \tr_{\H_A}( \rho_{AB}).
\end{align*}

The general bound from \cite{wolf} straightforwardly extends to the infinite-dimensional setting as follows.
\begin{lemma}[Boundary energy controls mutual informat.]
\label{th:lem1} Let $\beta > 0$ and suppose that $\tr e^{-\beta H_{AB} } < \infty$. Then
\[
I(A:B) \leq \beta \tr( H_{\partial} (\rho_A \otimes \rho_B - \rho_{AB} )) .
\]
\end{lemma}
\begin{proof} We start from
$$
\Fc^\beta ( \rho_{AB} )=F^\beta_{AB} \leq \Fc^\beta ( \rho_{A}\otimes \rho_B ).
$$
Using $S(\rho_A \otimes \rho_B ) = S(\rho_A) + S(\rho_B)$, we obtain
$$
I(A:B)\leq  \beta \tr( H_{AB} (\rho_A \otimes \rho_B - \rho_{AB} ))
$$
and the right-hand side equals $\beta \tr( H_{\partial} (\rho_A \otimes \rho_B - \rho_{AB} )) $ by basic properties of the partial trace.
\end{proof}
\section{Main Result}

Now we consider the Bose-Hubbard model in the framework of the previous section. Let $\Lambda\equiv \Lambda_L$ denote a box of side length $L$ in the $d$-dimensional lattice with periodic boundary conditions. For $x,y \in \Lambda_L$ we write $x \sim y$ if $x$ and $y$ are nearest neighbors in the periodized lattice. At each site lives a bosonic particle described by the local Hilbert space $\hh = \ell^2(\IN)$.

The total Hilbert space $\H_{AB} = \otimes_{x \in \Lambda_L} \hh$ is isomorphic to the Fock space  $\FF(\ell^2(\Lambda_L))$. On it, we consider the Bose-Hubbard Hamiltonian
\begin{align}
\label{eq:BH Hamiltonian}
H_{AB} = - J \sum_{x \sim y} a_x^\dagger a_y + \frac{U}{2} \sum_{x \in \Lambda_L} n_x (n_x -1) - \mu \N,
\end{align}
where $J\in\mathbb R$ represents the strength of the kinetic nearest-neighbor hopping, $U>0$ the strength of the on-site repulsion and $\mu \in\mathbb R$ the chemical potential, and $\N = \sum_{x \in \Lambda_L} n_x$ is the total number operator. The Hamiltonian is self-adjoint on a suitable domain $\Def(H_{AB})$; see e.g., \cite{faupin2022maximal} and \Cref{sec:operator domains}.



We are now ready to state the main result. 
We decompose the box into two regions $A$ and $B$ with boundary region
\begin{align*}
	\ABb = &\{ x \in A : \exists y \in B : x \sim y \} \\ &\cup \{ y \in B : \exists x \in A : x \sim y \}
\end{align*}
 as shown in Figure \ref{fig:model}.

\begin{thm}[Main result: thermal area law]
	\label{th:main}
	For all $\beta,U,\mu > 0$, we have
	\begin{align*}
		I(A:B) \leq c(J,U,\mu) \max\{1,\beta\} \abs{\ABb}
	\end{align*}
\end{thm}


A few remarks are in order: (i) The repulsiveness assumption $U>0$ is necessary as it ensures stability of the system. (ii) The assumption that $\mu>0$ is standard, cf.\ the usual phase diagram in Figure 13(a) in \cite{ultracold_gases_review}. Indeed, note that if one would take $\mu$ sufficiently negative, then the system becomes effectively devoid of particles in the grand-canonical setting. (iii) The constant $c(J,U,\mu)$ can be made explicit from \eqref{eq:main constant}. (iv) The maximum $\max\{1,\beta\}$ means that  the bound behaves as $\beta$ for low temperatures in accordance with \cite{wolf}. We recall that some growth in $\beta$ is strongly expected without any gap assumption \cite{gottesman2010entanglement}. For high temperature $\beta<1$, we find a temperature independent lower bound which matches the classical situation \cite{wolf}. (v) As explained in \cite{wolf}, a bound on the mutual information implies a bound for the correlation of any pair of bounded observables and the same is true in the bosonic setting, see \Cref{sec:truncated correlations}.

We close the presentation by discussing several extensions of the result which can be obtained from the same methods.
 The proof can be extended to Hamiltonians of the form
\[
H = H_0 + W,
\]
with
\begin{align}
	\nonumber
	H_0 &= \sum_{k=1}^M \sum_{l=1}^k  \sum_{ \substack{ \{x_1, \ldots, x_l, \\ y_1, \ldots, y_l\} \subset \Lambda_L}} \sum_{\nu_1 + \ldots + \nu_l = k} J^{\nu_1 \ldots \nu_l}_{x_1 \ldots x_l y_1 \ldots y_l} \\ &\qquad \times  (a^\dagger_{x_1})^{\nu_1} \ldots (a^\dagger_{x_l})^{\nu_l}  a_{y_1}^{\nu_1} \ldots  a_{y_l}^{\nu_l}, \label{eq:generalizations} \\
	W &= \sum_x f(n_x), \nonumber
\end{align}
where $M \in \IN$, $f \geq 0$ is a polynomially bounded function,  growing faster  than $x^M$, and  $J^{\nu_1 \ldots \nu_l}_{x_1 \ldots x_k y_1 \ldots y_k}$ is uniformly bounded and finite-range. Moreover, the proof also works if we add to the original Hamiltonian density-density interactions of the form
\begin{align*}
\widetilde H = \sum_{x,y}  J_{xy}	n_x n_y,
\end{align*}
provided that $U$ is sufficiently large, and $J_{xy}$ is uniformly bounded and finite-range. 
Furthermore, instead of finite-range interactions, we can consider hopping terms decaying at infinity fast enough, e.g.,  $-\sum_{x,y} J_{xy} a_x^\dagger a_y$ for $J_{xy} \geq 0$, $x \in \IZ^d$, satisfying $J_{xy}\leq (1+|x-y|)^{-\alpha}$ for sufficiently large $\alpha>d$. 

 Finally, the underlying lattice structure can be easily modified as well, though, the constant will be less explicit since it depends on the spectrum of the graph Laplacian.
  In summary, the thermal area law can be proved for an entire class of bosonic lattice gases in any dimension.

The translation-invariance of the underlying lattice is in fact necessary in our current proof. Dropping this assumption would require to control the number of bosons potentially accumulating on the boundaries of the whole system,  see also the proof of \cref{prop:step1}. It is an interesting open problem to remove the translation-invariance assumption. Furthermore, as we control the $H_0$ term with the on-site interaction $W$, we also need the rather strong decay assumption on $J_{xy}$. Therefore, an interesting problem is to develop an  alternative approach allowing for long-range interactions with $\alpha$ arbitrary close to $d$.

%

%
%
%
%
%
%
%
%
%

\section{Sketch of proof of Theorem \ref{th:main} }
The detailed proof of Theorem \ref{th:main} is given in \Cref{sec:proof}. Here we give a sketch of the main ideas.


We decompose the Hamiltonian as
$$
H_{AB} = H_0 + W + (\mu -2dJ) \N
$$
where $H_0 = -  J \sum_{x \sim y} a_x^\dagger a_y + 2dJ \N$ is the shifted kinetic term and $W = \frac{U}{2}  \sum_{x \in X} n_x (n_x -1)$ is the on-site interaction.

To use Lemma \ref{th:lem1}, we again decompose the Hamiltonian as  $H_{AB}=H_A+H_B+H_\partial$ where we define the subsystem Hamiltonians with open boundary conditions along the cut. More precisely, for $X \in \{A,B\}$, 
 we set
\begin{align*}
H_X =  -J \sum_{\substack{x \sim y, \\x,y \in X}} a_x^\dagger a_y + \frac{U}{2}  \sum_{x \in X} n_x (n_x -1) - \mu \sum_{x \in X} n_x 
\end{align*}
and $H_\partial = -J \sum_{\substack{x \sim y,\\x \in A, y \in B}} (a_x^\dagger a_y + a_y^\dagger a_x)$.

Notice that the boundary Hamiltonian $H_\partial$ contains at most $d \abs{\ABb}$ many summands, so the right-hand side of Lemma \ref{th:lem1} seems to exhibits the desired scaling in $L$ and $\beta$. The main challenge is that the hopping terms $a_x^\dagger a_y + a_y^\dagger a_x$ between $A$ and $B$ are unbounded in contrast to the cases of spin systems or lattice fermions.


Our first idea is that since expectations with respect to the full Gibbs state are rather difficult to handle, we aim for the expectation in a quasifree state which can be computed via Wick's rule. To this end, we want to remove the $W$ term in $e^{-\beta H_{AB}}$. The technical tool to rigorously implement such a shift in the operator exponent will be a double application of the Peierls-Bogoliubov inequality \cite[(2.14)]{carlen}, which has a long history in the study of quantum many-body systems and quantum information theory. Applying it twice, we obtain
\begin{align}
\label{eq:Peierls-Bogoliubov2}
\frac{\tr (P e^K)}{\tr (e^K)} \leq \frac{\tr (P e^ {K+P})}{\tr (e^{K+P})}.
\end{align}
We use this bound with $K=-\beta H_{AB}$ and $P = \beta (W - (\mu + \gamma  - 2dJ) \N )$ so that the new effective Hamiltonian is
$$
K+P=-\beta (H_0 + (\gamma+2dJ) \N)
$$
and $e^{-\beta (H_0 + (\gamma+2dJ) \N)}$ is a trace-class quasifree state, so that expectations can be calculated via Wick's rule. Here we introduced a parameter $\gamma>0$ large enough in order to make $e^{-\beta (H_0 + (\gamma+2dJ) \N)}$ normalizable. This can be interpreted as an \textit{artificial decrease of the chemical potential} which is introduced to stabilize the system by balancing the loss of the repulsion $W$ from the exponential. (Indeed, note that  without help from $\gamma$, we would get $e^{-\beta (H_0 - (\mu+2dJ) \N)}$ which has infinite trace for $\mu > -2dJ$.) 

The final expression that we arrive at via \eqref{eq:Peierls-Bogoliubov2} can then be evaluated using Wick's theorem. Subsequently, by means of the corresponding one-particle density operator, we obtain rather explicit expressions which amount to a Riemann sum of the density of  Planck's law. This density decays for large $\beta$ in an integrable way. In particular, we see that the resulting expression as well as the error terms are bounded for $\beta \geq 1$ and Theorem \ref{th:main} follows. For the details, see \Cref{sec:proof}. 

For the discussed generalizations \eqref{eq:generalizations}, one has to bound $H_0$ by $W$  and then remove the generalized hopping term $H_0$ instead of $W$ in the exponent by means of the Peierls-Bogoliubov argument. In the end, one obtains a trace just involving number operators, which can be easily computed as well.

%

\section{Conclusions}

We presented a rigorous proof for a thermal area law for the Bose-Hubbard model and related bosonic lattice gases. This result closes a gap in the recently growing literature about the quantum information theory of lattice bosons and provides a positive-temperature counterpart to the area law for gapped bosonic ground states (in 1D) \cite{abrahamsen2022entanglement}.

 The idea of the proof is based on the general idea in \cite{wolf}  together with a Peierls-Bogoliubov argument which artificially decreases the chemical potential in order to get a trace-class free reference state. The method is highly robust and extends to many other bosonic lattice gases and any lattice  dimension.


Natural follow-up problems include the approximability of bosonic thermal states by generalized matrix product operators in the spirit of \cite{mpo_arealaw_link}. This is related to area laws for the generalized Rényi entropy \cite{verstraete2006matrix}, so it would be interesting to extend the present results to some Rényi generalizations of the mutual information (as in \cite{renyi_positivetemperature}), see also \cite{gottesman2010entanglement}. Furthermore, in light of recent progress of Lieb-Robinson bounds for bosons \cite{schuch2011information,wang2020tightening,kuwahara2021lieb,faupin2022lieb,faupin2022maximal,yin2022finite,kuwahara2022optimal} one could explore the applicability of such bounds and their imaginary time counterparts in the context of thermal area laws, see also \cite{rapid_mixing2,improved} for connections between imaginary time Lieb-Robinson bounds and  thermal area laws.

On the one hand, our result is of fundamental nature in the quantum information theory of lattice systems. On the other hand, it paves the way for  future information-theoretic studies of lattice bosons such as the ones described above. The goal is to unlock the full potential of these experimentally finely tunable systems for modern applications such as quantum machine learning \cite{amin2018quantum,anshu2021sample} and semidefinite programming solvers \cite{brandao2017quantum,van2020quantum}.

\section*{Acknowledgments}
The authors are grateful to {\'A}lvaro Alhambra and {\'A}ngela Capel Cuevas for useful comments on a draft version of the manuscript.

\newcommand{\mt}{\widetilde \mu}

\printbibliography

@article{abrahamsen2019polynomial,
    author = {Abrahamsen, Nilin},
    journal = {arXiv:1907.04862},
    title = {A polynomial-time algorithm for ground states of spin trees},
    year = {2019},
    doi = {10.48550/arXiv.1907.04862}
}

@article{abrahamsen2022entanglement,
    author = {Abrahamsen, Nilin and Su, Yuan and Tong, Yu and Wiebe, Nathan},
    journal = {arXiv:2203.16012},
    title = {Entanglement area law for {1D} gauge theories and bosonic systems},
    year = {2022},
    doi = {10.48550/arXiv.2203.16012}
}

@article{alhambra_survey,
    author = {Alhambra, {\'A}lvaro M},
    journal = {arXiv:2204.08349},
    title = {Quantum many-body systems in thermal equilibrium},
    year = {2022},
    doi = {10.48550/arXiv.2204.08349}
}

@article{amin2018quantum,
    author = {Amin, Mohammad H and Andriyash, Evgeny and Rolfe, Jason and Kulchytskyy, Bohdan and Melko, Roger},
    doi = {10.1103/PhysRevX.8.021050},
    journal = {Phys. Rev. X},
    number = {2},
    pages = {021050},
    publisher = {APS},
    title = {Quantum {Boltzmann} machine},
    volume = {8},
    year = {2018}
}

@inproceedings{anshu2020entanglement,
    author = {Anshu, Anurag and Arad, Itai and Gosset, David},
    booktitle = {Proceedings of the 52nd Annual ACM SIGACT Symposium on Theory of Computing},
    doi = {10.1145/3357713.3384292},
    pages = {868--874},
    title = {Entanglement subvolume law for {2D} frustration-free spin systems},
    year = {2020}
}

@article{anshu2021sample,
    author = {Anshu, Anurag and Arunachalam, Srinivasan and Kuwahara, Tomotaka and Soleimanifar, Mehdi},
    doi = {10.1038/s41567-021-01232-0},
    journal = {Nat. Phys.},
    number = {8},
    pages = {931--935},
    publisher = {Nature Publishing Group},
    title = {Sample-efficient learning of interacting quantum systems},
    volume = {17},
    year = {2021}
}

@inproceedings{anshu2022area,
    author = {Anshu, Anurag and Arad, Itai and Gosset, David},
    booktitle = {Proceedings of the 54th Annual ACM SIGACT Symposium on Theory of Computing},
    doi = {10.1145/3519935.3519962},
    pages = {12--18},
    title = {An area law for 2d frustration-free spin systems},
    year = {2022}
}

@article{arad2012improved,
    author = {Arad, Itai and Landau, Zeph and Vazirani, Umesh},
    doi = {10.1103/PhysRevB.85.195145},
    journal = {Phys. Rev. B},
    number = {19},
    pages = {195145},
    publisher = {APS},
    title = {Improved one-dimensional area law for frustration-free systems},
    volume = {85},
    year = {2012}
}

@article{arad2013area,
    author = {Arad, Itai and Kitaev, Alexei and Landau, Zeph and Vazirani, Umesh},
    title = {An area law and sub-exponential algorithm for {1D} systems},
    year = {2013},
    doi = {10.48550/arXiv.1301.1162}
}

@article{arad2017rigorous,
    author = {Arad, Itai and Landau, Zeph and Vazirani, Umesh and Vidick, Thomas},
    doi = {10.1007/s00220-017-2973-z},
    journal = {Commun. Math. Phys.},
    number = {1},
    pages = {65--105},
    publisher = {Springer},
    title = {Rigorous {RG} algorithms and area laws for low energy eigenstates in {1D}},
    volume = {356},
    year = {2017}
}

@article{araki1969gibbs,
    author = {Araki, Huzihiro},
    doi = {10.1007/BF01645134},
    journal = {Commun. Math. Phys.},
    number = {2},
    pages = {120--157},
    publisher = {Springer},
    title = {Gibbs states of a one dimensional quantum lattice},
    volume = {14},
    year = {1969}
}

@article{audenaert2002entanglement,
    author = {Audenaert, K and Eisert, J and Plenio, MB and Werner, RF},
    doi = {10.1103/PhysRevA.66.042327},
    journal = {Phys. Rev. A},
    number = {4},
    pages = {042327},
    publisher = {APS},
    title = {Entanglement properties of the harmonic chain},
    volume = {66},
    year = {2002}
}

@article{barahona1982computational,
    author = {Barahona, Francisco},
    doi = {10.1088/0305-4470/15/10/028},
    journal = {J. Phys. A: Math. Gen.},
    number = {10},
    pages = {3241},
    publisher = {IOP Publishing},
    title = {On the computational complexity of Ising spin glass models},
    volume = {15},
    year = {1982}
}

@article{bh_es1,
    author = {Alba, Vincenzo and Haque, Masudul and L{\"a}uchli, Andreas M},
    doi = {10.1103/PhysRevLett.110.260403},
    journal = {Phys. Rev. Lett.},
    number = {26},
    pages = {260403},
    publisher = {APS},
    title = {Entanglement spectrum of the two-dimensional {Bose}-{Hubbard} model},
    volume = {110},
    year = {2013}
}

@article{bh_es2,
    author = {Metlitski, Max A and Grover, Tarun},
    journal = {arXiv:1112.5166},
    title = {Entanglement entropy of systems with spontaneously broken continuous symmetry},
    year = {2011},
    doi = {10.48550/arXiv.1112.5166}
}

@article{bh_es3,
    author = {Kallin, Ann B and Hastings, Matthew B and Melko, Roger G and Singh, Rajiv RP},
    doi = {10.1103/PhysRevB.84.165134},
    journal = {Phys. Rev. B},
    number = {16},
    pages = {165134},
    publisher = {APS},
    title = {Anomalies in the entanglement properties of the square-lattice {Heisenberg} model},
    volume = {84},
    year = {2011}
}

@article{bh_es4,
    author = {Song, H Francis and Laflorencie, Nicolas and Rachel, Stephan and Le Hur, Karyn},
    doi = {10.1103/PhysRevB.83.224410},
    journal = {Phys. Rev. B},
    number = {22},
    pages = {224410},
    publisher = {APS},
    title = {Entanglement entropy of the two-dimensional {Heisenberg} antiferromagnet},
    volume = {83},
    year = {2011}
}

@article{bosonic_coding,
    author = {Albert, Victor V},
    journal = {arXiv:2211.05714},
    title = {Bosonic coding: introduction and use cases},
    year = {2022},
    doi = {10.48550/arXiv.2211.05714}
}

@article{bousso2002holographic,
    author = {Bousso, Raphael},
    doi = {10.1007/978-94-010-0211-0_3},
    journal = {Rev. Mod. Phys.},
    number = {3},
    pages = {825},
    publisher = {APS},
    title = {The holographic principle},
    volume = {74},
    year = {2002}
}

@article{brandao2013area,
    author = {Brand{\~a}o, Fernando GSL and Horodecki, Michal},
    journal = {Nat. Phys.},
    number = {11},
    pages = {721--726},
    publisher = {Nature Publishing Group},
    title = {An area law for entanglement from exponential decay of correlations},
    volume = {9},
    year = {2013},
    doi = {10.1038/nphys2747}
}

@article{brandao2015entanglement,
    author = {Brandao, Fernando GSL and Cramer, Marcus},
    journal = {Phys. Rev. B},
    number = {11},
    pages = {115134},
    publisher = {APS},
    title = {Entanglement area law from specific heat capacity},
    volume = {92},
    year = {2015},
    doi = {10.1103/PhysRevB.92.115134}
}

@inproceedings{brandao2017quantum,
    author = {Brandao, Fernando GSL and Svore, Krysta M},
    booktitle = {2017 IEEE 58th Annual Symposium on Foundations of Computer Science (FOCS)},
    organization = {IEEE},
    pages = {415--426},
    title = {Quantum speed-ups for solving semidefinite programs},
    year = {2017},
    doi = {10.1109/FOCS.2017.45}
}

@book{bratellirobinsion2,
    author = {Bratteli, O. and Robinson, D.W.},
    doi = {10.1007/978-3-662-09089-3},
    isbn = {9783540614432},
    lccn = {gb97047698},
    publisher = {Springer Berlin Heidelberg},
    series = {Theoretical and Mathematical Physics},
    title = {Operator Algebras and Quantum Statistical Mechanics: Equilibrium States. Models in Quantum Statistical Mechanics},
    year = {2003}
}

@incollection{carlen,
    author = {Carlen, Eric},
    booktitle = {Entropy and the quantum},
    doi = {10.1090/conm/529/10428},
    mrclass = {47A63 (15A60 46L99 47N50 82B10 94A15)},
    mrnumber = {2681769},
    mrreviewer = {Rupert L. Frank},
    pages = {73--140},
    publisher = {Amer. Math. Soc., Providence, RI},
    series = {Contemp. Math.},
    title = {Trace inequalities and quantum entropy: an introductory
course},
    doi = {10.1090/conm/529/10428},
    volume = {529},
    year = {2010}
}

@inproceedings{childs2014bose,
    author = {Childs, Andrew M and Gosset, David and Webb, Zak},
    booktitle = {International Colloquium on Automata, Languages, and Programming},
    doi = {10.1007/978-3-662-43948-7_26},
    organization = {Springer},
    pages = {308--319},
    title = {The {Bose}-{Hubbard} model is {QMA}-complete},
    year = {2014}
}

@article{cho2014sufficient,
    author = {Cho, Jaeyoon},
    doi = {10.1103/PhysRevLett.113.197204},
    journal = {Phys. Rev. Lett.},
    number = {19},
    pages = {197204},
    publisher = {APS},
    title = {Sufficient condition for entanglement area laws in thermodynamically gapped spin systems},
    volume = {113},
    year = {2014}
}

@article{cho2018realistic,
    author = {Cho, Jaeyoon},
    doi = {10.1103/PhysRevX.8.031009},
    journal = {Phys. Rev. X},
    number = {3},
    pages = {031009},
    publisher = {APS},
    title = {Realistic area-law bound on entanglement from exponentially decaying correlations},
    volume = {8},
    year = {2018}
}

@article{convex_mps,
    author = {Berta, Mario and Brand\~ao, Fernando G. S. L. and Haegeman, Jutho and Scholz, Volkher B. and Verstraete, Frank},
    doi = {10.1103/PhysRevB.98.235154},
    issue = {23},
    journal = {Phys. Rev. B},
    month = {Dec},
    numpages = {8},
    pages = {235154},
    publisher = {American Physical Society},
    title = {Thermal states as convex combinations of matrix product states},
    url = {https://link.aps.org/doi/10.1103/PhysRevB.98.235154},
    volume = {98},
    year = {2018}
}

@article{cramer2006correlations,
    author = {Cramer, Marcus and Eisert, Jens},
    doi = {10.1088/1367-2630/8/5/071},
    journal = {New J. Phys.},
    number = {5},
    pages = {71},
    publisher = {IOP Publishing},
    title = {Correlations, spectral gap and entanglement in harmonic quantum systems on generic lattices},
    volume = {8},
    year = {2006}
}

@article{cramer2006entanglement,
    author = {Cramer, Marcus and Eisert, Jens and Plenio, Martin B and Dreissig, J},
    doi = {10.1103/PhysRevA.73.012309},
    journal = {Phys. Rev. A},
    number = {1},
    pages = {012309},
    publisher = {APS},
    title = {Entanglement-area law for general bosonic harmonic lattice systems},
    volume = {73},
    year = {2006}
}

@article{cramer2007statistics,
    author = {Cramer, Marcus and Eisert, Jens and Plenio, MB},
    doi = {10.1103/PhysRevLett.98.220603},
    journal = {Phys. Rev. Lett.},
    number = {22},
    pages = {220603},
    publisher = {APS},
    title = {Statistics dependence of the entanglement entropy},
    volume = {98},
    year = {2007}
}

@article{dalzell2019locally,
    author = {Dalzell, Alexander M and Brand{\~a}o, Fernando GSL},
    doi = {10.22331/q-2019-09-23-187},
    journal = {Quantum},
    pages = {187},
    publisher = {Verein zur F{\"o}rderung des Open Access Publizierens in den Quantenwissenschaften},
    title = {Locally accurate {MPS} approximations for ground states of one-dimensional gapped local {Hamiltonians}},
    volume = {3},
    year = {2019}
}

@article{de2010solving,
    author = {de Beaudrap, N and Ohliger, M and Osborne, TJ and Eisert, J},
    doi = {10.1103/PhysRevLett.105.060504},
    journal = {Phys. Rev. Lett.},
    number = {6},
    pages = {060504},
    publisher = {APS},
    title = {Solving frustration-free spin systems},
    volume = {105},
    year = {2010}
}

@article{dmrg_higher1,
    author = {Verstraete, Frank and Cirac, J Ignacio},
    journal = {arXiv:cond-mat/0407066},
    title = {Renormalization algorithms for quantum-many body systems in two and higher dimensions},
    year = {2004},
    doi = {10.48550/arXiv.cond-mat/0407066}
}

@article{dmrg_higher2,
    author = {Stoudenmire, Edwin M and White, Steven R},
    doi = {10.1146/annurev-conmatphys-020911-125018},
    journal = {Annu. Rev. Condens. Matter Phys.},
    number = {1},
    pages = {111--128},
    publisher = {Annual Reviews},
    title = {Studying two-dimensional systems with the density matrix renormalization group},
    volume = {3},
    year = {2012}
}

@article{dmrg_higher3,
    author = {Schollw{\"o}ck, Ulrich},
    doi = {10.1098/rsta.2010.0382},
    journal = {Philos. Trans. Royal Soc. A },
    number = {1946},
    pages = {2643--2661},
    publisher = {The Royal Society Publishing},
    title = {The density-matrix renormalization group: a short introduction},
    volume = {369},
    year = {2011}
}

@article{eisert2010colloquium,
    author = {Eisert, Jens and Cramer, Marcus and Plenio, Martin B},
    doi = {10.1103/RevModPhys.82.277},
    journal = {Rev. Mod. Phys.},
    number = {1},
    pages = {277},
    publisher = {APS},
    title = {Colloquium: Area laws for the entanglement entropy},
    volume = {82},
    year = {2010}
}

@article{entanglement_negativity,
    author = {Sherman, Nicholas E and Devakul, Trithep and Hastings, Matthew B and Singh, Rajiv RP},
    doi = {10.1103/PhysRevE.93.022128},
    journal = {Phys. Rev. E.},
    number = {2},
    pages = {022128},
    publisher = {APS},
    title = {Nonzero-temperature entanglement negativity of quantum spin models: Area law, linked cluster expansions, and sudden death},
    volume = {93},
    year = {2016}
}

@book{entropy_book,
    author = {Ohya, Masanori and Petz, D{\'e}nes},
    doi = {10.1007/978-3-642-57997-4},
    publisher = {Springer Science \& Business Media},
    title = {Quantum entropy and its use},
    year = {2004}
}

@article{faupin2022lieb,
    author = {Faupin, J{\'e}r{\'e}my and Lemm, Marius and Sigal, Israel Michael},
    doi = {10.1007/s00220-022-04416-8},
    journal = {Commun. Math. Phys.},
    number = {3},
    pages = {1011--1037},
    publisher = {Springer},
    title = {On {Lieb}-{Robinson} Bounds for the {Bose}-{Hubbard} Model},
    volume = {394},
    year = {2022}
}

@article{faupin2022maximal,
    author = {Faupin, J{\'e}r{\'e}my and Lemm, Marius and Sigal, Israel Michael},
    doi = {10.1103/PhysRevLett.128.150602},
    journal = {Phys. Rev. Lett.},
    number = {15},
    pages = {150602},
    publisher = {APS},
    title = {Maximal speed for macroscopic particle transport in the {Bose}-{Hubbard} model},
    volume = {128},
    year = {2022}
}

@article{free_fermions,
    author = {H Bernigau and M J Kastoryano and J Eisert},
    doi = {10.1088/1742-5468/2015/02/p02008},
    journal = {J. Stat. Mech: Theory Exp.},
    month = {feb},
    number = {2},
    pages = {P02008},
    publisher = {{IOP} Publishing},
    title = {Mutual information area laws for thermal free fermions},
    url = {https://doi.org/10.1088/1742-5468/2015/02/p02008},
    volume = {2015},
    year = {2015}
}

@article{frohlich2015some,
    author = {Fr{\"o}hlich, J{\"u}rg and Ueltschi, Daniel},
    doi = {10.1063/1.4921305},
    journal = {J. Math. Phys.},
    number = {5},
    pages = {053302},
    publisher = {AIP Publishing LLC},
    title = {Some properties of correlations of quantum lattice systems in thermal equilibrium},
    volume = {56},
    year = {2015}
}

@article{ge2016area,
    author = {Ge, Yimin and Eisert, Jens},
    doi = {10.1088/1367-2630/18/8/083026},
    journal = {New J. Phys.},
    number = {8},
    pages = {083026},
    publisher = {IOP Publishing},
    title = {Area laws and efficient descriptions of quantum many-body states},
    volume = {18},
    year = {2016}
}

@article{goldberg2015complexity,
    author = {Goldberg, Leslie Ann and Jerrum, Mark},
    doi = {10.1073/pnas.1505664112},
    journal = {Proc. Natl. Acad. Sci. U.S.A.},
    number = {43},
    pages = {13161--13166},
    publisher = {National Acad Sciences},
    title = {A complexity classification of spin systems with an external field},
    volume = {112},
    year = {2015}
}

@article{gottesman2010entanglement,
    author = {Gottesman, Daniel and Hastings, Matthew B},
    doi = {10.1088/1367-2630/12/2/025002},
    journal = {New J. Phys.},
    number = {2},
    pages = {025002},
    publisher = {IOP Publishing},
    title = {Entanglement versus gap for one-dimensional spin systems},
    volume = {12},
    year = {2010}
}

@article{gross1979decay,
    author = {Gross, Leonard},
    doi = {10.1007/BF01562538},
    journal = {Commun. Math. Phys.},
    number = {1},
    pages = {9--27},
    publisher = {Springer},
    title = {Decay of correlations in classical lattice models at high temperature},
    volume = {68},
    year = {1979}
}

@article{hastings2007entropy,
    author = {Hastings, Matthew B},
    doi = {10.1103/PhysRevB.76.035114},
    journal = {Phys. Rev. B},
    number = {3},
    pages = {035114},
    publisher = {APS},
    title = {Entropy and entanglement in quantum ground states},
    volume = {76},
    year = {2007}
}

@article{hastings_arealaw,
    author = {Hastings, Matthew B},
    doi = {10.1088/1742-5468/2007/08/P08024},
    journal = {J. Stat. Mech: Theory Exp.},
    number = {08},
    pages = {P08024},
    publisher = {IOP Publishing},
    title = {An area law for one-dimensional quantum systems},
    volume = {2007},
    year = {2007}
}

@article{huang2014area,
    author = {Huang, Yichen},
    journal = {arXiv:1403.0327},
    title = {Area law in one dimension: Degenerate ground states and {Renyi} entanglement entropy},
    year = {2014},
    doi = {10.48550/arXiv.1403.0327}
}

@article{improved,
    author = {Kuwahara, Tomotaka and Alhambra, {\'A}lvaro M and Anshu, Anurag},
    doi = {10.1103/PhysRevX.11.011047},
    journal = {Phys. Rev. X},
    number = {1},
    pages = {011047},
    publisher = {APS},
    title = {Improved thermal area law and quasilinear time algorithm for quantum {Gibbs} states},
    volume = {11},
    year = {2021}
}

@article{kato2019quantum,
    author = {Kato, Kohtaro and Brandao, Fernando GSL},
    doi = {10.1007/s00220-019-03485-6},
    journal = {Commun. Math. Phys.},
    number = {1},
    pages = {117--149},
    publisher = {Springer},
    title = {Quantum approximate Markov chains are thermal},
    volume = {370},
    year = {2019}
}

@article{kliesch2014locality,
    author = {Kliesch, Martin and Gogolin, Christian and Kastoryano, MJ and Riera, A and Eisert, J},
    doi = {10.1103/PhysRevX.4.031019},
    journal = {Phys. Rev. X},
    number = {3},
    pages = {031019},
    publisher = {APS},
    title = {Locality of temperature},
    volume = {4},
    year = {2014}
}

@article{kuwahara2020area,
    author = {Kuwahara, Tomotaka and Saito, Keiji},
    doi = {10.1038/s41467-020-18055-x},
    journal = {Nat. Commun.},
    number = {1},
    pages = {1--7},
    publisher = {Nature Publishing Group},
    title = {Area law of noncritical ground states in {1D} long-range interacting systems},
    volume = {11},
    year = {2020}
}

@article{kuwahara2020clustering,
    author = {Kuwahara, Tomotaka and Kato, Kohtaro and Brandao, Fernando GSL},
    doi = {10.1103/PhysRevLett.124.220601},
    journal = {Phys. Rev. Lett.},
    number = {22},
    pages = {220601},
    publisher = {APS},
    title = {Clustering of conditional mutual information for quantum Gibbs states above a threshold temperature},
    volume = {124},
    year = {2020}
}

@article{kuwahara2020gaussian,
    author = {Kuwahara, Tomotaka and Saito, Keiji},
    doi = {10.1016/j.aop.2020.168278},
    journal = {Ann. Phys.},
    pages = {168278},
    publisher = {Elsevier},
    title = {Gaussian concentration bound and ensemble equivalence in generic quantum many-body systems including long-range interactions},
    volume = {421},
    year = {2020}
}

@article{kuwahara2021lieb,
    author = {Kuwahara, Tomotaka and Saito, Keiji},
    doi = {10.1103/PhysRevLett.127.070403},
    journal = {Phys. Rev. Lett.},
    number = {7},
    pages = {070403},
    publisher = {APS},
    title = {Lieb-{Robinson} bound and almost-linear light cone in interacting boson systems},
    volume = {127},
    year = {2021}
}

@article{kuwahara2022optimal,
    author = {Kuwahara, Tomotaka and Van Vu, Tan and Saito, Keiji},
    journal = {arXiv:2206.14736},
    title = {Optimal light cone and digital quantum simulation of interacting bosons},
    year = {2022},
    doi = {10.48550/arXiv.2206.14736}
}

@article{lenci2005large,
    author = {Lenci, Marco and Rey-Bellet, Luc},
    doi = {10.1007/s10955-005-3015-3},
    journal = {J. Stat. Phys.},
    number = {3},
    pages = {715--746},
    publisher = {Springer},
    title = {Large deviations in quantum lattice systems: one-phase region},
    volume = {119},
    year = {2005}
}

@article{lindblad,
    author = {Lindblad, G{\"o}ran},
    doi = {10.1007/BF01609396},
    journal = {Commun. Math. Phys.},
    number = {2},
    pages = {147--151},
    publisher = {Springer},
    title = {Completely positive maps and entropy inequalities},
    volume = {40},
    year = {1975}
}

@article{longrange2,
    author = {Gong, Zhe-Xuan and Foss-Feig, Michael and Brand\~ao, Fernando G. S. L. and Gorshkov, Alexey V.},
    doi = {10.1103/PhysRevLett.119.050501},
    issue = {5},
    journal = {Phys. Rev. Lett.},
    month = {Jul},
    numpages = {6},
    pages = {050501},
    publisher = {American Physical Society},
    title = {Entanglement Area Laws for Long-Range Interacting Systems},
    url = {https://link.aps.org/doi/10.1103/PhysRevLett.119.050501},
    volume = {119},
    year = {2017}
}

@article{marien2016entanglement,
    author = {Mari{\"e}n, Micha{\"e}l and Audenaert, Koenraad MR and Van Acoleyen, Karel and Verstraete, Frank},
    doi = {10.1007/s00220-016-2709-5},
    journal = {Commun. Math. Phys.},
    number = {1},
    pages = {35--73},
    publisher = {Springer},
    title = {Entanglement rates and the stability of the area law for the entanglement entropy},
    volume = {346},
    year = {2016}
}

@article{masanes2009area,
    author = {Masanes, Llu{\'\i}s},
    doi = {10.1103/PhysRevA.80.052104},
    journal = {Phys. Rev. A},
    number = {5},
    pages = {052104},
    publisher = {APS},
    title = {Area law for the entropy of low-energy states},
    volume = {80},
    year = {2009}
}

@article{michalakis2012stability,
    author = {Michalakis, Spyridon},
    journal = {arXiv:1206.6900},
    title = {Stability of the area law for the entropy of entanglement},
    year = {2012},
    doi = {10.48550/arXiv.1206.6900}
}

@article{mpo_arealaw_link,
    author = {Guth Jarkovský, Jiří and Moln\'ar, Andr\'as and Schuch, Norbert and Cirac, J. Ignacio},
    doi = {10.1103/PRXQuantum.1.010304},
    issue = {1},
    journal = {PRX Quantum},
    month = {Sep},
    numpages = {5},
    pages = {010304},
    publisher = {American Physical Society},
    title = {Efficient Description of Many-Body Systems with Matrix Product Density Operators},
    url = {https://link.aps.org/doi/10.1103/PRXQuantum.1.010304},
    volume = {1},
    year = {2020}
}

@article{netovcny2004large,
    author = {Neto{\v{c}}n{\`y}, K and Redig, F},
    journal = {J. Stat. Phys.},
    number = {3},
    pages = {521--547},
    publisher = {Springer},
    title = {Large deviations for quantum spin systems},
    volume = {117},
    year = {2004},
    doi = {10.1007/s10955-004-3452-4}
}

@misc{nielsen2002quantum,
    author = {Nielsen, Michael A and Chuang, Isaac},
    doi = {10.1119/1.1463744},
    publisher = {American Association of Physics Teachers},
    title = {Quantum computation and quantum information},
    year = {2002}
}

@article{nisq1,
    author = {Motta, Mario and Sun, Chong and Tan, Adrian TK and O’Rourke, Matthew J and Ye, Erika and Minnich, Austin J and Brand{\~a}o, Fernando GSL and Chan, Garnet Kin},
    doi = {10.1038/s41567-020-0798-8},
    journal = {Nat. Phys.},
    number = {2},
    pages = {205--210},
    publisher = {Nature Publishing Group},
    title = {Determining eigenstates and thermal states on a quantum computer using quantum imaginary time evolution},
    volume = {16},
    year = {2020}
}

@article{nisq2,
    author = {Lamm, Henry and Lawrence, Scott},
    doi = {10.1103/PhysRevLett.121.170501},
    journal = {Phys. Rev. Lett.},
    number = {17},
    pages = {170501},
    publisher = {APS},
    title = {Simulation of nonequilibrium dynamics on a quantum computer},
    volume = {121},
    year = {2018}
}

@article{nisq3,
    author = {Beach, Matthew JS and Melko, Roger G and Grover, Tarun and Hsieh, Timothy H},
    doi = {10.1103/PhysRevB.100.094434},
    journal = {Phys. Rev. B},
    number = {9},
    pages = {094434},
    publisher = {APS},
    title = {Making trotters sprint: A variational imaginary time ansatz for quantum many-body systems},
    volume = {100},
    year = {2019}
}

@article{nisq4,
    author = {Yuan, Xiao and Endo, Suguru and Zhao, Qi and Li, Ying and Benjamin, Simon C},
    doi = {10.22331/q-2019-10-07-191},
    journal = {Quantum},
    pages = {191},
    publisher = {Verein zur F{\"o}rderung des Open Access Publizierens in den Quantenwissenschaften},
    title = {Theory of variational quantum simulation},
    volume = {3},
    year = {2019}
}

@article{nisq5,
    author = {McArdle, Sam and Jones, Tyson and Endo, Suguru and Li, Ying and Benjamin, Simon C and Yuan, Xiao},
    doi = {10.1038/s41534-019-0187-2},
    journal = {Npj Quantum Inf.},
    number = {1},
    pages = {1--6},
    publisher = {Nature Publishing Group},
    title = {Variational ansatz-based quantum simulation of imaginary time evolution},
    volume = {5},
    year = {2019}
}

@article{nisq6,
    author = {Yeter-Aydeniz, K{\"u}bra and Pooser, Raphael C and Siopsis, George},
    doi = {10.1038/s41534-020-00290-1},
    journal = {Npj Quantum Inf.},
    number = {1},
    pages = {1--8},
    publisher = {Nature Publishing Group},
    title = {Practical quantum computation of chemical and nuclear energy levels using quantum imaginary time evolution and Lanczos algorithms},
    volume = {6},
    year = {2020}
}

@article{nisq7,
    author = {Love, Peter J},
    doi = {10.1038/s41567-019-0709-z},
    journal = {Nat. Phys.},
    number = {2},
    pages = {130--131},
    publisher = {Nature Publishing Group},
    title = {Cooling with imaginary time},
    volume = {16},
    year = {2020}
}

@article{park1995uniqueness,
    author = {Park, Yong Moon and Yoo, Hyun Jae},
    doi = {10.1007/BF02178359},
    journal = {J. Stat. Phys.},
    number = {1},
    pages = {223--271},
    publisher = {Springer},
    title = {Uniqueness and clustering properties of Gibbs states for classical and quantum unbounded spin systems},
    volume = {80},
    year = {1995}
}

@article{plenio2005entropy,
    author = {Plenio, Martin B and Eisert, Jens and Dreissig, J and Cramer, Marcus},
    doi = {10.1103/PhysRevLett.94.060503},
    journal = {Phys. Rev. Lett.},
    number = {6},
    pages = {060503},
    publisher = {APS},
    title = {Entropy, entanglement, and area: analytical results for harmonic lattice systems},
    volume = {94},
    year = {2005}
}

@article{PRXQuantum.2.040331,
    author = {Alhambra, \'Alvaro M. and Cirac, J. Ignacio},
    doi = {10.1103/PRXQuantum.2.040331},
    issue = {4},
    journal = {PRX Quantum},
    month = {Nov},
    numpages = {18},
    pages = {040331},
    publisher = {American Physical Society},
    title = {Locally Accurate Tensor Networks for Thermal States and Time Evolution},
    url = {https://link.aps.org/doi/10.1103/PRXQuantum.2.040331},
    volume = {2},
    year = {2021}
}

@article{qec,
   author={Steven M. Girvin},
   journal={SciPost Phys. Lect. Notes},
   pages={70},
   year={2023},
   publisher={SciPost},
    doi = {10.21468/SciPostPhysLectNotes.70},
    title = {Introduction to quantum error correction and fault tolerance}
}

@article{rapid_mixing1,
    author = {Kastoryano, Michael J and Eisert, Jens},
    doi = {10.1063/1.4822481},
    journal = {J. Math. Phys.},
    number = {10},
    pages = {102201},
    publisher = {American Institute of Physics},
    title = {Rapid mixing implies exponential decay of correlations},
    volume = {54},
    year = {2013}
}

@article{rapid_mixing2,
    author = {Brandao, Fernando GSL and Cubitt, Toby S and Lucia, Angelo and Michalakis, Spyridon and Perez-Garcia, David},
    journal = {J. Math. Phys.},
    number = {10},
    pages = {102202},
    publisher = {AIP Publishing LLC},
    title = {Area law for fixed points of rapidly mixing dissipative quantum systems},
    volume = {56},
    year = {2015},
    doi = {10.1063/1.4932612}
}

@article{renyi_positivetemperature,
    author = {Scalet, Samuel O. and Alhambra, {\'{A}}lvaro M. and Styliaris, Georgios and Cirac, J. Ignacio},
    issn = {2521-327X},
    journal = {{Quantum}},
    month = {September},
    pages = {541},
    publisher = {{Verein zur F{\"{o}}rderung des Open Access Publizierens in den Quantenwissenschaften}},
    title = {Computable {R}{\'{e}}nyi mutual information: {A}rea laws and correlations},
    volume = {5},
    year = {2021},
    doi = {10.22331/q-2021-09-14-541}
}

@book{ruelle,
    address = {},
    author = {Ruelle, David},
    doi = {10.1142/4090},
    edition = {},
    publisher = {New York: W.A. Benjamin},
    title = {Statistical Mechanics: Rigorous Results},
    year = {1969}
}

@article{ruskai1972inequalities,
    author = {Ruskai, Mary Beth},
    doi = {10.1007/BF01645523},
    journal = {Commun. Math. Phys.},
    number = {4},
    pages = {280--289},
    publisher = {Springer},
    title = {Inequalities for traces on von Neumann algebras},
    volume = {26},
    year = {1972}
}

@article{schuch2011information,
    author = {Schuch, Norbert and Harrison, Sarah K and Osborne, Tobias J and Eisert, Jens},
    doi = {10.1103/PhysRevA.84.032309},
    journal = {Phys. Rev. A},
    number = {3},
    pages = {032309},
    publisher = {APS},
    title = {Information propagation for interacting-particle systems},
    volume = {84},
    year = {2011}
}

@article{tajik2022experimental,
	title={Verification of the area law of mutual information in a quantum field simulator},
	author={Tajik, Mohammadamin and Kukuljan, Ivan and Sotiriadis, Spyros and Rauer, Bernhard and Schweigler, Thomas and Cataldini, Federica and Sabino, Jo{\~a}o and M{\o}ller, Frederik and Sch{\"u}ttelkopf, Philipp and Ji, Si-Cong and others},
	journal={Nat. Phys.},
	pages={1--5},
	year={2023},
	publisher={Nature Publishing Group UK London},
	doi = {10.1038/s41567-023-02027-1}
}

@article{thermal_numerics,
    author = {{\v{Z}}nidari{\v{c}}, Marko and Prosen, Toma{\v{z}} and Pi{\v{z}}orn, Iztok},
    journal = {Phys. Rev. A},
    number = {2},
    pages = {022103},
    publisher = {APS},
    title = {Complexity of thermal states in quantum spin chains},
    volume = {78},
    year = {2008},
    doi = {10.1103/PhysRevA.78.022103}
}

@article{tong2021provably,
    author = {Tong, Yu and Albert, Victor V and McClean, Jarrod R and Preskill, John and Su, Yuan},
    doi = {10.22331/q-2022-09-22-816},
    journal = {Quantum},
    pages = {816},
    publisher = {Verein zur F{\"o}rderung des Open Access Publizierens in den Quantenwissenschaften},
    title = {Provably accurate simulation of gauge theories and bosonic systems},
    volume = {6},
    year = {2022}
}

@article{ueltschi2004,
    author = {D. Ueltschi},
    doi = {10.17323/1609-4514-2004-4-2-511-522},
    journal = {Moscow Math. J.},
    number = {2},
    pages = {511–522},
    title = {Cluster Expansions and Correlation Functions},
    volume = {4},
    year = {2004}
}

@article{ultracold_gases_review,
    author = {Bloch, Immanuel and Dalibard, Jean and Zwerger, Wilhelm},
    doi = {10.1103/RevModPhys.80.885},
    journal = {Rev. Mod. Phys.},
    number = {3},
    pages = {885},
    publisher = {APS},
    title = {Many-body physics with ultracold gases},
    volume = {80},
    year = {2008}
}

@article{van2013entanglement,
    author = {Van Acoleyen, Karel and Mari{\"e}n, Micha{\"e}l and Verstraete, Frank},
    doi = {10.1103/PhysRevLett.111.170501},
    journal = {Phys. Rev. Lett.},
    number = {17},
    pages = {170501},
    publisher = {APS},
    title = {Entanglement rates and area laws},
    volume = {111},
    year = {2013}
}

@article{van2020quantum,
    author = {Van Apeldoorn, Joran and Gily{\'e}n, Andr{\'a}s and Gribling, Sander and de Wolf, Ronald},
    doi = {10.1109/FOCS.2017.44},
    journal = {Quantum},
    pages = {230},
    publisher = {Verein zur F{\"o}rderung des Open Access Publizierens in den Quantenwissenschaften},
    title = {Quantum {SDP}-solvers: Better upper and lower bounds},
    volume = {4},
    year = {2020}
}

@article{verstraete2006matrix,
    author = {Verstraete, Frank and Cirac, J Ignacio},
    doi = {10.1103/PhysRevB.73.094423},
    journal = {Phys. Rev. B},
    number = {9},
    pages = {094423},
    publisher = {APS},
    title = {Matrix product states represent ground states faithfully},
    volume = {73},
    year = {2006}
}

@article{wang2020tightening,
    author = {Wang, Zhiyuan and Hazzard, Kaden RA},
    doi = {10.1103/PRXQuantum.1.010303},
    journal = {PRX Quantum},
    number = {1},
    pages = {010303},
    publisher = {APS},
    title = {Tightening the {Lieb}-{Robinson} bound in locally interacting systems},
    volume = {1},
    year = {2020}
}

@article{white1993density,
    author = {White, Steven R},
    doi = {10.1103/PhysRevB.48.10345},
    journal = {Phys. Rev. B},
    number = {14},
    pages = {10345},
    publisher = {APS},
    title = {Density-matrix algorithms for quantum renormalization groups},
    volume = {48},
    year = {1993}
}

@article{wolf,
    author = {Wolf, Michael M and Verstraete, Frank and Hastings, Matthew B and Cirac, J Ignacio},
    doi = {10.1103/PhysRevLett.100.070502},
    journal = {Phys. Rev. Lett.},
    number = {7},
    pages = {070502},
    publisher = {APS},
    title = {Area laws in quantum systems: mutual information and correlations},
    volume = {100},
    year = {2008}
}

@article{yin2022finite,
    author = {Yin, Chao and Lucas, Andrew},
    doi = {10.1103/PhysRevX.12.021039},
    journal = {Phys. Rev. X},
    number = {2},
    pages = {021039},
    publisher = {APS},
    title = {Finite speed of quantum information in models of interacting bosons at finite density},
    volume = {12},
    year = {2022}
}

\widetext
 
%
%

\appendix

 \section{Preliminaries}
\subsection{Operator domains}
\label{sec:operator domains}
For the general setup, we require the following statements about operator domains. We assume that $H_A$ and $H_B$ are densely defined symmetric operators on $\Def(A)$ and $\Def(B)$, respectively and the boundary Hamiltonian is defined on $H_\partial$ on $\Def(A) \otimes \Def(B)$. We assume that $H_{AB}$ is essentially self-adjoint on $\Def(A) \otimes \Def(B)$ which is then the appropriate domain for the equality
\begin{align}
H_{AB} = H_A \otimes \Id  + \Id \otimes H_B + H_\partial
\end{align}

Concerning the Bose-Hubbard model, in the total Fock space $\FF(\ell^2(\Lambda_L))$, we consider the dense domain
$$
\FF_\fin( \ell^2(\Lambda_L) ) = \{ \psi \in \FF( \ell^2(\Lambda_L) ) : \exists n_0 \in \IN : \forall n \geq n_0 : \psi_n = 0  \}.$$
The operator $H_{AB}$ is self-adjoint on the largest domain $\Def(H_{AB})$, where it can be defined, cf.\ \cite{faupin2022maximal}. It then follows from $[H_{AB},\N] = 0$ that $H_{AB}$ is indeed essentially self-adjoint on $\FF_\fin( \ell^2(\Lambda_L) )$. Moreover, for subsystems $X \in \{A,B\}$, we always set $\Def(X) = \FF_\fin(\H_X)$. The above properties can then be verified.


\subsection{Trace inequalities in infinite dimensions}

\begin{lemma}[Peierls-Bogoliubov inequality]
\label{th:Peierls-Bogoliubov}
Let $\N \geq 0$ be self-adjoint operator with purely discrete spectrum, let $\Pi_N := \ind_{\N \leq N}$, $N \in \IN$ and assume that $\Pi_N \H$ is finite-dimensional for all $N \in \IN$. 
Let $(\Def(K),K)$ be a self-adjoint and $(\Def(P),P)$ be a symmetric operator such that $[K,\N] = 0$, $[P,\N] = 0$, $K+P$ is self-adjoint, and $e^K$, $P e^K$, $e^{K+P}$ are all trace-class.
Then we have
\begin{align}
\label{eq:Peierls-Bogoliubov1}
\frac{\tr (P e^K)}{\tr (e^K)} \leq \log \left( \frac{\tr( e^{K+P} )}{\tr(e^K) } \right),
\end{align}
In particular, if $P e^ {K+P}$ is trace-class as well,  we obtain \eqref{eq:Peierls-Bogoliubov2}, i.e.,
\begin{align}
\frac{\tr (P e^K)}{\tr (e^K)} \leq \frac{\tr (P e^ {K+P})}{\tr (e^{K+P})}.
\end{align}
\end{lemma}
\begin{proof}[Proof of \cref{th:Peierls-Bogoliubov}]
The inequality \eqref{eq:Peierls-Bogoliubov1} is well-known for matrices, see for example \cite[(2.14)]{carlen} (or more generally if $P$ is bounded \cite{ruskai1972inequalities}). Therefore, we have 
\begin{align*}
\frac{\tr (P e^K \Pi_N)}{\tr (e^K \Pi_N)} \leq \log \left( \frac{\tr( e^{K+P} \Pi_N )}{\tr(e^K \Pi_N) } \right),
\end{align*}
where $\Pi_N = \ind_{\N \leq N}$. Taking the limit $N \to \infty$ yields the desired result. Finally, \eqref{eq:Peierls-Bogoliubov2} follows from \eqref{eq:Peierls-Bogoliubov1} by means of
\[
 \log \left( \frac{\tr( e^{K+P} )}{\tr(e^K) } \right) = -  \log \left( \frac{\tr(e^K)}{  \tr( e^{K+P} )} \right)  \leq \frac{\tr (P e^ {K+P})}{\tr (e^{K+P})}. 
\] This proves \cref{th:Peierls-Bogoliubov}.
\end{proof}
\begin{prop}[Gibbs variational principle]
\label{th:passivity}
Let $(H,\Def(H))$ be a self-adjoint operator such that $\tr e^{-\beta H} < \infty$ for all $\beta > 0$. Let  $\Fc^\beta ( \rho)  :=     \left( \tr( H \rho) -  \frac{S(\rho)}{\beta} \right)$. Then
\[
\inf_{\rho \in \Sp{\H}} \Fc^\beta(\rho) = \Fc^\beta(\rho_H),
\]
where $\rho_H = e^{-\beta H} / \tr e^{-\beta H}$. 
\end{prop}
\begin{proof}[Proof of \cref{th:passivity}]
We have
\[
\tr(K \ln K - K \ln P) \geq \tr(K-P)
\]
for all positive self-adjoint trace class operators $K,P$ \cite[Prop. 2.5.3]{ruelle}. This yields \cite[Section 5.3.1]{bratellirobinsion2}
\begin{align*}
\Fc^\beta(\rho) &= \beta^{-1} \tr( \rho \ln \rho  - \rho \ln \rho_H) - \beta^{-1}\ln \tr e^{-\beta H} \geq -\beta^{-1}\ln \tr e^{-\beta H} =  \Fc^\beta(\rho_H). \qedhere
\end{align*}
\end{proof}

%
%
%
%


\subsection{Truncated correlations bound}
\label{sec:truncated correlations}

Let $M_A,M_B$ be two bounded self-adjoint operators on $\H_A$ and $\H_B$, respectively. We denote their truncated correlation function as
\[
\mathcal{C}(A,B) = \tr( M_A \otimes M_B \rho_{AB}) - \tr( M_A \rho_A ) \tr( M_B \rho_B ).
\]
For all $\beta,U,\mu > 0$, we have
\begin{align}
	\label{eq:correlations bound}
	\frac{\mathcal{C}(A,B)^2}{2\nn{M_A}^2 \nn{M_B}^2} \leq 2c(J,U,\mu) \max\{1,\beta\} \abs{\ABb},
\end{align}
with the same constant as in \cref{th:main}. This follows from  the quantum Pinsker inequality $I(A:B) \geq \frac{1}{2} \nn{ \rho_{AB} - \rho_A \otimes \rho_B }_1^2$ and $\nn{X}_1 \geq \tr(XY) / \nn{Y}$. We mention that the standard proof of the quantum Pinkser inequality, cf.\ \cite[Theorem 1.15]{entropy_book} extends to infinite dimensions since the data processing inequality holds for trace-class operators \cite{lindblad}.

 Note however that \eqref{eq:correlations bound} is trivial unless $\abs{\ABb}$ stays bounded in the infinite-volume limit, e.g. if $A$ or $B$ are kept at fixed size, as its left-hand side is always bounded by one. Therefore, in order to find more regimes where this is useful, it is an interesting open question if our main estimate in \cref{th:main} can be improved for $\beta \to 0$, as it is the case for spin systems \cite{wolf,improved}.
%

 \section{Proof of Theorem \ref{th:main}}
 \label{sec:proof}
 To begin, we notice that we may assume without loss of generality that $J>0$. Indeed, if $J=0$, the Gibbs state is a product state (Mott insulator) and the claim is trivial and if $J<0$ we can employ the unitary transformation $a_x\to -a_x$ at every second lattice site to reduce to the case $J>0$.
 
 \subsection{Step 1: Controlling boundary energy by particle number}
In this section, we prove the following bound
\begin{prop}\label{prop:step1}
	\label{th:final IAB bound}
	We have
	\[
	I(A:B)  \leq \frac{4 d \abs{\ABb}}{L^d} \beta J  \tr ( \N \rho_{AB} ) .
	\]
\end{prop}
\begin{proof}[Proof of Proposition \ref{prop:step1}]
		Using \cref{th:lem1} and the operator Cauchy-Schwarz inequality $\pm (a_x^\dagger a_y + a_y^\dagger a_x) \leq n_x + n_y$ we get 
	\begin{align*}
	I(A:B) \leq \beta J \sum_{\substack{x \sim y, \\ x \in A, y \in B}}  \tr( (n_x+n_y) (\rho_A \otimes \rho_B + \rho_{AB} )) .
	\end{align*}
	Since, for $x\in A$, $ \tr( n_x \rho_A \otimes \rho_B ) = \tr( n_x \rho_A) = \tr( n_x \rho_{AB})$ and similarly for $n_y$ and $\rho_B$, we obtain
		\begin{align}
	I&(A:B) \leq  2 \beta J \sum_{\substack{x \sim y, \\ x \in A, y \in B}} \tr( (n_x+n_y)\rho_{AB} ).  	\label{eq:lem2}
	\end{align}
	Observe that $H_{AB}$ is translation-invariant, i.e., for all $x \in \Lambda_L$, \[\T_x^* H_{AB} \T_x = H_{AB},\] where $\T_x$ denotes the unitary translation operator by $x$,  $(\T_x \psi)(y) = \psi(y+x \mod L)$. This implies $\T_x^* \rho_{AB} \T_x = \rho_{AB}$ and therefore, $\tr (n_x \rho_{AB}) = \tr ( n_y \rho_{AB})$ for all $x,y$. The summation in \eqref{eq:lem2} is over $\abs{\ABb}$ many terms, so
	\begin{align*}
	\sum_{x \sim y, ~x \in A, y \in B}  \tr( (n_x+n_y)\rho_{AB} )  \leq   2 d \abs{\ABb}  \tr( n_{x_0} \rho_{AB}) 
	=  2d  \abs{\ABb}  \frac{\tr( \N \rho_{AB} )}{L^d} ,
	\end{align*}
	where $x_0 \in \Lambda$ is some arbitrary element. This proves \cref{prop:step1}.
\end{proof}

\subsection{Step 2: Removing the interaction from the exponential}
In the following we write $\expec{P}{K} := \tr(P e^{-K}) / \tr(e^{-K})$ for operators $K$ and $P$ and we will also drop the subscript $AB$ and write $H = H_{AB}$.

\begin{prop}\label{prop:step2}
\label{th:N estimate}
Let $C_0 = \frac{1}{4} \frac{U}{\gamma + \mu }$. For all $\gamma >0$, we have
\begin{align*}
\expec{\N}{\beta H} \leq 2  \bigg(  \frac{C_0}{U/2}   \expec{W - (\mu + \gamma - 2dJ) \N } {\beta (H_0 +\gamma \N)}   +  \frac{C_0}{4} (1+ C_0^{-1})^2 L^d \bigg).
\end{align*}
\end{prop}
\begin{proof}[Proof of Proposition \ref{prop:step2}]
	Let $C > 0$.
We have $n \leq  C n(n-1) + \frac{C}{4} (1+ C^{-1})^2$ for all $n \in \IN$. Thus, we obtain
\begin{align*}
\expec{\N}{\beta H} = \expec{\sum_{x \in \Lambda_L} n_x }{\beta H} 
\leq  \sum_{x \in \Lambda_L} \expec{  C n_x(n_x-1) +  \frac{C}{4} (1+ C^{-1})^2}{\beta H} =  \frac{C}{U/2} \expec{W}{\beta H}+  \frac{C}{4} (1+ C^{-1})^2 L^d.
\end{align*}
This leads to
\begin{align*}
\expec{\N}{\beta H}  \left(1 - \frac{C}{U/2} (\mu + \gamma ) \right) &\leq \frac{C}{U/2} \expec{W - (\mu + \gamma  ) \N}{\beta H} +  \frac{C}{4} (1+ C^{-1})^2 L^d
\\ &\leq \frac{C}{U/2} \expec{W - (\mu + \gamma - 2dJ  ) \N}{\beta H} +  \frac{C}{4} (1+ C^{-1})^2 L^d.
\end{align*}
We now apply twice the Peierls-Bogoliubov inequality, i.e., \eqref{eq:Peierls-Bogoliubov2} in \cref{th:Peierls-Bogoliubov} with $P = \beta (W - (\mu + \gamma  - 2dJ) \N )$ and $K = -\beta H$, and get 
\begin{align*}
\expec{W - (\mu + \gamma - 2dJ) \N}{\beta H}  \leq  \expec{W - (\mu + \gamma - 2dJ) \N}{\beta (H_0 +\gamma \N)} .
\end{align*}
Using this in the previous bound and setting $C = C_0$ yields the desired estimate. This proves Proposition \ref{prop:step2}.
\end{proof}

\subsection{Step 3: Calculation for quasi-free states}
The estimate of the on-site interaction in the free reference states leads to an estimate of the particle number on a specific site via the one-particle density matrix. Here we get a Riemann sum of the density function of Planck's law. 

We set
\[
 f(\gamma, \beta, J)  := \int_{[0,\frac{1}{2}]^d}  \left(e^{4 J\beta \sum_{j=1}^d  \sin^2 ( \pi x_j ) + \beta \gamma } - 1 \right)^{-1} \d x,
\]
and denote the error terms by
\begin{align*}
\epsilon_1 &=  \frac{(L+1)^d}{L^d} -1, \\
\epsilon_2 &= \frac{(e^{\beta \gamma} -1)^{-1}}{L^d} \left( (L+2)^d - (L+1)^d  \right) .
\end{align*}

We shall use the following two lemmas.
\begin{lemma}
\label{th:a_xdagger a_x estimate}
For all $\gamma >0$ and $x \in \Lambda_L$,
\begin{align*}
\expec{a_x^\dagger a_x}{\beta (H_0 + \gamma \N)} \leq  & 2^d (1 + \epsilon_1)   f(\gamma, \beta, J)  + \epsilon_2.
\end{align*}
%
\end{lemma}
 \begin{lemma}
	\label{th:f estimate}
	For all $\gamma > 0$,
	\[
	f(\gamma, \beta, J)  \leq \frac{1}{2^d} \frac{1}{ \beta \gamma}.
	\]
\end{lemma}

We postpone the proofs of these lemmas for now and show how they imply they main result.

\begin{proof}[Proof of \cref{th:main}]
 Wick's theorem for quasi-free states \cite[p.40 and Prop. 5.2.28]{bratellirobinsion2} yields
\begin{align*}
\expec{ n_x ( n_x - 1)}{{\beta (H_0 +\gamma \N)}  } = \expec{ (a_x^\dagger a_x)^2}{\beta (H_0 +\gamma \N)} = 2 \expec{ a_x^\dagger a_x}{\beta (H_0 +\gamma \N)}^2.
\end{align*}

By \cref{th:a_xdagger a_x estimate},
\begin{align*}
\expec{ W }{\beta (H_0 +\gamma \N)}  = 2  L^d \expec{ a_x^\dagger a_x}{\beta (H_0 +\gamma \N)}^2 \leq  2^{2d +2} L^d (1+\epsilon_1)^2  f(\gamma,\beta,J)^2 + 4 L^d \epsilon_2^2.
\end{align*}
Using this and $\expec{- (\mu + \gamma - 2dJ) \N } {\beta (H_0 +\gamma \N)} \leq 0$ in the upper bound of \cref{th:N estimate}, we obtain that for all $\gamma > 2dJ - \mu$, 
\begin{align*}
& \expec{\N}{\beta H} \leq 2 L^d   \bigg( \frac{C_0}{U/2} (  2^{2d +2} (1+\epsilon_1)^2  f(\gamma,\beta,J)^2 + 4 \epsilon_2^2 ) +  \frac{C_0}{4} (1+ C_0^{-1})^2  \bigg).
\end{align*}
Next we use that $f(\gamma, \beta, J)  \leq \frac{1}{2^d} \frac{e^{- \alpha \beta \gamma}}{(1-\alpha) \beta \gamma}$ for all $\alpha \in (0,1)$, cf. \cref{th:f estimate}. The upper bound for $\expec{\N}{\beta H}$ then becomes
\begin{align*}
2 L^d  \bigg(  \frac{1}{2} \frac{1}{\gamma + \mu}\left( 4 (1 +\epsilon_1)^2 \frac{1}{ \beta^2 \gamma^2} + 4  \epsilon_2^2 \right)  + \frac{1}{16} \frac{U}{\gamma  + \mu} + \frac{\gamma  + \mu}{U}      + \frac{1}{2}  \bigg).
\end{align*}
Estimating $1 + \epsilon_1 \leq 2^d$ and $\epsilon_2 \leq 3^d (e^{\beta \gamma }  -  1 )^{-1}$ and  choosing $\gamma = \max\{1/\beta,  2dJ + 1\}$ (such that $\beta \gamma \geq 1$), we arrive at the bound
\begin{align}
\expec{\N}{\beta H} &\leq  	2 L^d  \bigg(  \frac{1}{2( 2dJ + 1 + \mu)}\bigg(4\cdot 4^d  + 4 \cdot 3^{2d} (e -  1 )^{-2}  \bigg)  + \frac{1}{16} \frac{U}{ 2dJ + 1 + \mu} + \frac{ \max\{1/\beta,  2dJ + 1\} + \mu}{U}      + \frac{1}{2}   \bigg) \nonumber \\
	&\leq \frac{L^d}{8} \max\left\lbrace \frac{1}{\beta} ,1 \right\rbrace c(J,U,\mu), \label{eq:last step}
\end{align}
with 
\begin{align}
	c(J,U,\mu) := 16 d J \bigg( \frac{ 4^d +  4 \cdot 3^{2d} (e  -  1 )^{-2}  }{2(2dJ + 1 + \mu)} + \frac{1}{16} \frac{U}{ 2dJ + 1 + \mu} + \frac{  2dJ + 1 + \mu}{U}      + \frac{1}{2}  \bigg),   \label{eq:main constant}
\end{align}
and where the last step \eqref{eq:last step} follows from distinguishing the cases $\beta \gtreqless 1$. Finally, we conclude the proof with \cref{th:final IAB bound}.
\end{proof}
\begin{remark}
Notice that the error terms $\epsilon_1$ and $\epsilon_2$ converge to zero as $L \to \infty$. It is of separate interest whether the remaining term $2^d   f(\gamma,\beta,J)$ actually captures the correct asymptotic behavior of $\expec{a_x^\dagger a_x}{\beta (H_0 + \gamma \N)}$, cf. \cref{th:a_xdagger a_x estimate}. 
\end{remark}

\subsection{Laplace eigenvalues and eigenvectors for periodic boundary conditions}
The proofs proof of Lemma \ref{th:a_xdagger a_x estimate} requires information on the spectral theory of the one-body graph Laplacian, for which introduce notation here.

Consider the discrete Laplacian $-\Delta \geq 0$ on the chain $\{0,1, \ldots, L\}$, $L \in \IN$, with periodic boundary conditions, i.e., assume that the nodes $0$ and $L$ are identified such that every vector on $\ell^2(\{0,1, \ldots, L\})$ satisfies $v(0) = v(L)$. So it suffices to use vectors $v \in \ell^2(\{1,\ldots,L\})$. 

The $L$ eigenvalues $\lambda_i$ and eigenvectors $v_i$, $i=1, \ldots, L$ of $-\Delta$ in this situation are given by
\begin{align*}
\lambda_{2k+1} &= 4 \sin^2 \left( \frac{k \pi}{L} \right), ~ k = 0, \ldots,  \left\lfloor \frac{L-1}{2} \right\rfloor,  \\
\lambda_{2k} &= 4 \sin^2 \left( \frac{k \pi}{L} \right) , ~  k = 1, \ldots,  \left\lfloor \frac{L}{2} \right\rfloor,
\end{align*}
and
\begin{align*}
v_1(i) &= L^{-1/2}, \\
v_L(i) &= L^{-1/2} (-1)^i \text{ if } L \text{ even}, \\
v_{2k + 1}(i) &= \sqrt{ 2 / L} \cos \left( \frac{ \pi k }{L} (2 i - 1)   \right) ,  ~ k = 1, \ldots,  \left\lfloor \frac{L-1}{2} \right\rfloor, \\
v_{2k}(i) &= \sqrt{ 2 / L} \sin \left( \frac{ \pi k }{L} (2 i - 1)   \right), ~ k = 1, \ldots,  \left\lfloor \frac{L-1}{2} \right\rfloor.
\end{align*}

\subsection{Proofs of Lemmas \ref{th:a_xdagger a_x estimate} and \ref{th:f estimate}}

In this section, we give the still outstanding proofs of Lemmas \ref{th:a_xdagger a_x estimate} and \ref{th:f estimate}.

\begin{proof}[Proof of \cref{th:a_xdagger a_x estimate}]
	Let $x_0 = (1,\ldots,1) \in \Lambda_L$.
	By the formula for the one-particle density matrix, cf. e.g. \cite[Prop. 5.2.28]{bratellirobinsion2} and by translation-invariance we find
	\begin{align*}
		\expec{a_x^\dagger a_x}{\beta (H_0 + \gamma \N)}  &= \sc{ \delta_x, e^{-\beta\gamma} e^{-\beta(- J \Delta)} (\Id - e^{-\beta\gamma} e^{-\beta (- J\Delta)} )^{-1} \delta_x} \\
		&= \sc{ \delta_{x_0},   (e^{\beta  (-J \Delta+\gamma)}  -  \Id )^{-1} \delta_{x_0}} \\
		&= \sum_{i_1,\ldots,i_d=1}^L   (e^{\beta ( J \sum_{j=1}^d   \lambda_{i_j} + \gamma)}  -  1 )^{-1} \prod_{j=1}^{d}  \abs{  v_{i_j}(1) }^2.
	\end{align*}
	Let $\ell = \lceil \frac{L}{2} \rceil$. Then we have for any $c > 0$ 
	\begin{align}
		\sum_{i=1}^{L} ( e^{\beta (J \lambda_i +c)} -1)^{-1} \abs{v_i(1)}^2 
		&= \frac{1}{L} (e^{\beta c}  -  1 )^{-1}   +  \sum_{i=1}^\ell (e^{\beta (J \lambda_{2i} +c)}  -  1 )^{-1} \bigg(   \abs{  v_{2i}(1) }^2  
		 + \ind_{L = 2\ell+1} \abs{  v_ {2i+1}(1) }^2 \bigg)  \nonumber   \\
		&\leq \frac{1}{L} (e^{\beta c}  -  1 )^{-1}   + \frac{2}{L}  \sum_{i=1}^\ell   (e^{\beta (J \lambda_{2i} +c)}  -  1 )^{-1} .  \label{eq:use for induction}
	\end{align}
	Then, by induction over $d$ and using \eqref{eq:use for induction} for the induction step, we find as an upper bound
	\begin{align}
		\expec{a_x^\dagger a_x}{\beta (H_0 + \gamma \N)} \leq    \frac{1}{L^d}\sum_{k=0}^{d} \binom{d}{k} 2^k  \sum_{i_1, \ldots, i_k=1}^\ell  (e^{\beta ( J \sum_{j=1}^k \lambda_{2i_j} +\gamma)}  -  1 )^{-1}, \label{eq:riemann sum}
	\end{align}
	where the sum over the $i_1, \ldots, i_k$ is defined to be one if $k=0$. 
	This can be proven like the binomial theorem, more precisely, 
\begin{align*}
	\sum_{i_1,\ldots,i_{d+1}=1}^L   &(e^{\beta ( J \sum_{j=1}^{d+1}   \lambda_{i_j} + \gamma)}  -  1 )^{-1} \prod_{j=1}^{d+1}  \abs{  v_{i_j}(1) }^2 \\ &\leq \frac{1}{L^d}\sum_{k=0}^{d} \binom{d}{k} 2^k  \sum_{i_1, \ldots, i_k=1}^\ell  \sum_{i_{d+1}=1}^L (e^{\beta ( J \sum_{j=1}^k \lambda_{2i_j} + \lambda_{i_{d+1}} +\gamma)}  -  1 )^{-1}   \abs{  v_{i_{d+1}}(1) }^2  \\
	&\leq 	\frac{1}{L^{d+1}}\sum_{k=0}^{d} \binom{d}{k} 2^k  \sum_{i_1, \ldots, i_k=1}^\ell  \left(  (e^{\beta ( J \sum_{j=1}^k \lambda_{2i_j} +\gamma)}  -  1 )^{-1}   + 2 \sum_{i_{d+1}=1}^\ell   (e^{\beta ( J \sum_{j=1}^{k+1} \lambda_{2i_j} +\gamma)}  -  1 )^{-1} \right) \\
	&= \frac{1}{L^{d+1}}\sum_{k=0}^{d+1} \binom{d+1}{k} 2^k  \sum_{i_1, \ldots, i_k=1}^\ell  (e^{\beta ( J \sum_{j=1}^k \lambda_{2i_j} +\gamma)}  -  1 )^{-1},
\end{align*}
where we use the induction hypothesis in the first step and \eqref{eq:use for induction}  for the second step. 
	With $2\ell-1 \leq L$ the first terms for $k=0, \ldots, d-1$ in \eqref{eq:riemann sum}  can be estimated by
	\begin{align*}
		\frac{(e^{\beta \gamma} -1)^{-1}}{L^d} \sum_{k=0}^{d-1} \binom{d}{k} (2 \ell)^k &\leq 
		\frac{(e^{\beta \gamma} -1)^{-1}}{L^d} \sum_{k=0}^{d-1} \binom{d}{k} (L+1)^k =
		 \frac{(e^{\beta \gamma} -1)^{-1}}{L^d} \left( (L+2)^d - (L+1)^d  \right) 
		  = \epsilon_2.
	\end{align*}
	%
	The term for $k=d$ in \eqref{eq:riemann sum} represents a Riemann sum of a function, decreasing in each argument, which is evaluated at the minimal points of the hypercubes $[(i_1-1) \pi/2, i_1 \pi/2] \times \ldots \times [(i_d-1) \pi/2, i_d \pi/2]$, i.e.,
	\begin{align*}
		&\frac{2^d}{L^d}  \sum_{i_1, \ldots, i_d=1}^\ell  (e^{\beta ( J \sum_{j=1}^d \lambda_{2i_j} +\gamma)}  -  1 )^{-1}  \\
		&= \frac{(2\ell)^d}{L^d \ell^d}   \sum_{i_1, \ldots, i_d=1}^\ell  (e^{\beta ( 4J \sum_{j=1}^d  \sin^2( i_j \pi / L) +\gamma)}  -  1 )^{-1}  \\ 
		&\leq  \frac{(L+1)^d}{L^d \ell^d}    \sum_{i_1, \ldots, i_d=1}^\ell  (e^{\beta ( 4J \sum_{j=1}^d  \sin^2( i_j \pi / (2\ell)) +\gamma)}  -  1 )^{-1}  \\ 
		&\leq   \frac{(L+1)^d}{L^d} \int_{[0,1]^d} (e^{\beta ( 4J \sum_{j=1}^d  \sin^2( x_j \pi / 2 ) +\gamma)}  -  1 )^{-1} \d x \\
		&=  2^d \frac{(L+1)^d}{L^d}  f(\gamma, \beta, J). \qedhere
	\end{align*}
\end{proof}

\begin{proof}[Proof of \cref{th:f estimate}]
	We have for any $\alpha \in (0,1)$,
	\begin{align*}
		&\int_{[0,\frac{1}{2}]^d}  (e^{4J\beta \sum_{j=1}^d  \sin^2 ( \pi x_j ) +\beta \gamma } - 1)^{-1} \d x \\
		&= e^{- \alpha \beta \gamma } \int_{[0,\frac{1}{2}]^d} (e^{4J\beta\sum_{j=1}^d  \sin^2 ( \pi x_j ) + (1-\alpha) \beta \gamma } - e^{- \alpha \beta \gamma })^{-1} \d x  \\
		&\leq e^{- \alpha \beta \gamma } \int_{[0,\frac{1}{2}]^d} (e^{4J\beta\sum_{j=1}^d  \sin^2 ( \pi x_j ) + (1-\alpha) \beta \gamma } - 1)^{-1} \d x  \\
		&\leq e^{- \alpha \beta \gamma } \int_{[0,\frac{1}{2}]^d}  ( (1-\alpha) \beta \gamma)^{-1} \d x  = \frac{1}{2^d} \frac{e^{-\alpha \beta \gamma}}{(1-\alpha) \beta \gamma}.
	\end{align*}
	Then taking $\alpha \to 0$ proves \cref{th:f estimate}.
\end{proof}

%
%

\end{document}